\documentclass[12pt]{amsart}
\usepackage[T1]{fontenc}
\usepackage{amsthm}
\usepackage{amssymb}
\usepackage{amsmath}
\usepackage{amsfonts}
\usepackage{mathtools}
\usepackage{thmtools}
\usepackage{lmodern}

\usepackage{xcolor}
\usepackage{tikz}
\usepackage{scalefnt}
\usepackage{indentfirst}
\usepackage{float}

\usepackage[a4paper, left=3cm, right=3cm, top=3cm, bottom=3cm,dvips]{geometry}
\usepackage[font=small,labelfont=bf]{caption}
\usepackage{enumitem}

\newtheorem{theorem}{Theorem}[section]
\newtheorem{lemma}[theorem]{Lemma}

\newtheorem{corollary}[theorem]{Corollary}

\newcommand{\N}{{\mathbb N}}

\newcommand{\sspace}{\mathcal{S}}
\newcommand{\qc}{Q} 

\title{On multidimensional generalization of binary search}
\author{Dariusz Dereniowski}
\address{Faculty of Electronics, Telecommunications and Informatics, Gdansk University of Technology, Gda\'{n}sk, Poland}
\email{deren@eti.pg.edu.pl}

\author{Przemys\l{}aw Gordinowicz}
\address{Institute of Mathematics, Lodz University of Technology, \L{}\'od\'z, Poland}
\email{pgordin@p.lodz.pl}

\author{Karolina Wr\'obel}
\address{Institute of Mathematics, Lodz University of Technology, \L{}\'od\'z, Poland}
\email{karolina.anna.wrobel1@gmail.com}

\begin{document}

\begin{abstract}
This work generalizes the binary search problem to a $d$-dimensional domain $S_1\times\cdots\times S_d$, where $S_i=\{0, 1, \ldots,n_i-1\}$ and $d\geq 1$, in the following way. Given $(t_1,\ldots,t_d)$, the target element to be found, the result of a comparison of a selected element $(x_1,\ldots,x_d)$ is the sequence of inequalities each stating that either $t_i < x_i$ or $t_i>x_i$, for $i\in\{1,\ldots,d\}$, for which at least one is correct, and the algorithm does not know the coordinate $i$ on which the correct direction to the target is given.
Among other cases, we show asymptotically almost matching lower and upper bounds of the query complexity to be in $\Omega(n^{d-1}/d)$ and $O(n^d)$ for the case of $n_i=n$.
In particular, for fixed $d$ these bounds asymptotically do match.
This problem is equivalent to the classical binary search in case of one dimension and shows interesting differences for higher dimensions.
For example, if one would impose that each of the $d$ inequalities is correct, then the search can be completed in $\log_2\max\{n_1,\ldots,n_d\}$ queries.
In an intermediate model when the algorithm knows which one of the inequalities is correct 
the sufficient number of queries is $\log_2(n_1\cdot\ldots\cdot n_d)$.
The latter follows from a graph search model proposed by Emamjomeh-Zadeh et al. [STOC 2016].
\end{abstract}

\maketitle

\section{Introduction} \label{wstep}

One of many applications of the classical binary search is to find a given value in a linear order representing potential outcomes of some kind of experimental measures (see~\cite{Ben-OrH08} for further references).
In such case, a comparison corresponds to performing one measure which shall reveal if the desired value is smaller or greater than the threshold picked for the measure.
A possible multi-dimensional extension of such search is when each particular experimental measure gives directions regarding many properties.
To provide a specific example consider a process of preparing a perfect drink.
Measurement of a test drink gives information if it is sweet enough and sour enough, which exemplifies a two-dimensional search case.
Note that learning that a given drink is too sour and not sweet enough provides information that in the search of the perfect drink one may exclude those that are even more sour and less sweet.
However, adding sugar and reducing acid is not guaranteed to be the correct way to improve the recipe.
For example, it is possible that the best drink could be produced by reducing the amount of acid significantly and slightly reducing the amount of sugar.
This dynamics is in line with our formulation of the search.

\subsection{Problem definition} \label{sec:problem-formulation}

One way to define the problem (see below for further comments) is to see it as an adaptive query-reply search game between two players called \emph{Algorithm} and \emph{Adversary}. 
Consider a $d$-dimensional search space (domain) $\sspace=S_1\times\cdots\times S_d$, where $S_i=\{0, 1, \ldots,n_i-1\}$, $d\geq 1$ and each $n_i$ is a positive integer.
One element $t=(t_1,\ldots,t_d)\in\sspace$ is called the \emph{target}.
The target is picked in advance by Adversary and is unknown to Algorithm.
That is, Algorithm may deduce $t$ only by considering the history of the game.

The game is divided into \emph{rounds}, where in each round Algorithm makes one \emph{query} and Adversary gives a \emph{reply}. The query is simply one element $q=(q_1,\ldots,q_d)$ of $\sspace$.
The reply is the following.
If $t = q$, then Adversary informs that the target is found, otherwise the reply gives for each $i\in\{1,\ldots,d\}$ information whether $t_i < q_i$ or $t_i > q_i$.
In the latter case, Adversary is bound by a restriction that for at least one such $i$ the given inequality holds.
(Hence, for the one dimensional case of $d=1$ our problem is exactly the classical binary search in a sorted array).
We note that Algorithm does not know the $i$, i.e., it does not learn the dimension $i$ on which the correct direction to the target is given in each reply.

The goal of Algorithm is to determine the target using the minimum number of queries.
We refer to this measure as the \emph{query complexity} and denote it by $\qc(n_1,\ldots,n_d)$.
When $n_1 = n_2 = \dots = n_d = n$, we denote $\qc(n^{\times d}) = \qc(n_1,\dots n_d)$.

Besides a formal statement of our problem, we also give some remarks regarding alternative formulations.
To this end we introduce the following notion of compatibility.
Consider a reply to a query on $q=(q_1,\ldots,q_d)$.
Consider any element $u=(u_1,\ldots,u_d)$ such that there exists $i$ for which either the reply says $t_i < q_i$ and it holds $u_i < q_i$, or the reply says $t_i>q_i$ and it holds $u_i>q_i$.
Then we say that $u$ is \emph{compatible} with the reply.
Each point that is compatible with a reply is a candidate for the target from the point of view of Algorithm.
Thus, any point that is not compatible is certainly not the target and hence Algorithm may exclude it from further considerations.
Then, one may equivalently state the problem so that Adversary does not need to pick the target in advance but instead maintains the set of points that are compatible with all replies to date, calling such points \emph{potential targets}.
The goal of Adversary is to ensure that the set of potential targets has more than one point for as long as possible.
This places the problem in the area of perfect information games.
Moreover, when one requires for a given search space, the output to be an optimal search strategy (often called a decision tree), then the problem becomes non-adaptive, purely combinatorial.

\subsection{Our results} \label{sec:our-results}
Here we formally state all results which provide a series of upper and lower bounds. The first bound in Theorem~\ref{thm:grid-tight} (the case of 2-dimensional grids) comes from Lemmas~\ref{lem:grid-lower} and~\ref{lem:grid-upper}.
\begin{theorem} \label{thm:grid-tight}
For any $m \geq n \geq 1$, 
\[\qc(m,n) \leq 2 n \left( \log_2 \frac{m}{n + 1} + 4 \right).\]
This bound is asymptotically tight.
\end{theorem}
Then we move to the 3-dimensional case, where the asymptotically exact query complexity is due to the next theorem and Corollary~\ref{cor:upper3g}.
\begin{theorem} \label{thm:lower3D}
For each $n_1\geq n_2\geq n_3\geq 2$ it holds
\[\qc(n_1, n_2, n_3) \geq \frac{1}{2} \big( n_2 - 1 \big) n_3 \left( \log_2 \frac{n_1 - 1}{n_2 - 1} + 1 \right).\]
\end{theorem}

For an arbitrary dimension $d$ we obtain the two following bounds.
\begin{theorem} \label{thm:lower-cube} 
For $d \geq 3$ and $n \geq 2$ it holds 
\[\qc(n^{\times d}) \geq 2 \left \lfloor \frac{n^{d - 1}}{d - 1}. \right \rfloor\]
\end{theorem}
\begin{theorem} \label{thm:upper-d}
For any $d\geq 2$ and any $n_1\geq\cdots\geq n_d\geq 1$ it holds
\[\qc(n_1, n_2, \dots, n_d) \in O \left( \prod_{i = 2}^{d} n_i \left( \log_2 \frac{n_1}{n_{2}} + 1 \right) \right).\]
\end{theorem}

Note that this generic bound when applied to the case of $n_i=n$ for each $i$, it gives an almost matching counterpart to Theorem~\ref{thm:lower-cube}:
\begin{corollary} \label{cor:upper-cube}
For any $d\geq 2$ and any $n\geq 2$, $\qc(n ^{\times d}) \in O( n^{d-1} ).$
\end{corollary}
This leaves our first open question: what is the query complexity $\qc(n^{\times d})$ for an arbitrary dimension $d$?
We note that Theorem~\ref{thm:upper-d} applied for $d=3$ gives an upper bound that is asymptotically matching Theorem~\ref{thm:lower3D}: 
\begin{corollary} \label{cor:upper3g}
For each $n_1\geq n_2\geq n_3\geq 1$, it holds
\[\qc(n_1, n_2, n_3) \in O \left( n_2 n_3 \left( \log_2 \frac{n_1}{n_2} + 1 \right) \right).\]
\end{corollary}

\subsection{Related work} \label{sec:related-work}

The study of the classical binary search has a long and rich history.
Since the problem in its basic setting is well understood, a lot of work has been done towards generalizations.
One of those includes considering noise, i.e., allowing that Adversary not always gives a correct reply \cite{Ben-OrH08,DLU21,EKMS20,FeigeRPU94,KarpK07}.
Another that is closer to our setting is so called group testing with its many variations.
In one of them the goal is to identify many targets and a response to a test understood as a subset of the search space is positive if it contains at least one target, see e.g. \cite{DyachkovVPS16}.
Particularly, the search space can be a continuous unit square \cite{BonisGV97}.
As an interesting example we mention the case where one searches for two targets within disjoint set of items of size $n$ and $m$, and each test pick an arbitrary subset with a reply indicating whether at least one target is in the subset; in this case it turns out that $\log_2(mn)$ queries are enough \cite{ChangH80}.
We refer interested reader to surveys outlining different models, connections to information theory and applications \cite{Deppe2007,Pelc02}.

The first works on graph searching have been phrased as searching partial orders \cite{Ben-AsherFN99,LamY01,MozesOW08,OnakP06}.
(Searching a tree-like partial order with a maximum element can be rephrased as edge search in a tree, which is equivalent to the edge ranking problem \cite{Dereniowski08} and consequently to other graph parameters like tree-depth \cite{NesetrilM06}).
Also, weighted versions, again most intensively studied, have been considered.
In the weighted versions, each point has an associated query cost (or duration) and the goal is to minimize the total cost \cite{CicaleseJLV12,CicaleseKLPV16,DereniowskiKUZ17}.
A generalization to graphs has been introduced in \cite{EKS16}: each query picks a point $q$ in a graph and Adversary informs Algorithm that $q$ is the target if that is the case, or otherwise Adversary gives any edge incident to $q$ that lies on the shortest path between $q$ and the target.
It turns out that for any $n$-point graph $\log_2n$ queries are enough \cite{EKS16}.
We mention interesting applications of the graph searching in machine learning \cite{Emamjomeh-ZadehK17}.

We note that the majority of works focus solely on the query complexity.
However, from the point of view of potential applications, it may be of interest to have an algorithm whose computational complexity per query is low.
We refer to some works that address this issue \cite{DeligkasMS19,DereniowskiLU21,EKS16}.

Further works on graph setting include models like edge queries and pair queries \cite{DereniowskiGP23}, searching for multiple targets \cite{DeligkasMS19}, or performing queries by a mobile agent \cite{BoczkowskiFKR21}.

\section{Preliminaries}

A reply that informs that the target has been found will be called a \emph{yes-reply}.
In our analysis we are mainly interested in handling replies that do not end the search, i.e, those that do not inform that the queried point is the target.
The game may be considered as a perfect information game without yes-replies in the following way. Let $P_k$ be the set of potential targets, after $k$th round. Initially, before the first query $P_0 = \sspace$.
Now, suppose that in the $k$th round Algorithm queried a point $q = (q_1, \dots, q_d) \in \sspace$.
The reply may be viewed as a binary sequence $r = (r_1, \dots r_d)$, such that $r_i = 1$ when the reply says $t_i < q_i$ or $r_i = 0$, when $t_i > q_i$.
By de Morgan's law the set of points inconsistent with the reply is $X = X_1 \times X_2 \times \dots \times X_d$, where $X_i = \{\min\{q_i, r_i(n_i-1)\},\ldots,\max\{q_i, r_i(n_i-1)\}\}$.
Hence, the set of potential targets after $k$th round is $P_k = P_{k-1} \setminus X.$ Moreover, because a sequence $r$ may be identified with the corner of the grid $\sspace$, we say that Adversary (negatively) \emph{answered to the corner} $r$.
We point out that the binary sequences are not extremities of the search space $\sspace$.
For example, in a grid $\sspace=S_1\times S_2$, for a query on a $q=(q_1,q_2)$ a reply saying $t_1<q_1$ and $t_2<q_2$ is interpreted as answering to the corner $(1,1)$ regardless of the values of $|S_1|$ and $|S_2|$.
This notation simplifies our statements regarding directions of reply inequalities by eliminating the need of referring to the dimension sizes.

We will use the following Jensen type inequality:
\begin{theorem} \label{sum_of_logs}
For any $n \in \mathbb{N}$ and any sequence $\left( x_1, x_2, \dots, x_n \right)$ of positive numbers
\begin{displaymath}
\sum_{i=1}^n \log_2 x_i \leq n \log_2 \left( \frac{\sum_{i=1}^n x_i}{n} \right).
\end{displaymath}
\end{theorem}


\section{Optimal solution on $2$ dimensions}

In this section we show bounds for the query complexity for a $2$-dimensional grid.
We hence prove Theorem~\ref{thm:grid-tight} by giving the respective lower and upper bounds in Lemmas \ref{lem:grid-lower} and \ref{lem:grid-upper}.

\begin{lemma} \label{lem:grid-lower}
For any $m \geq n \geq 1$  it holds $\qc(m,n)\geq n\log_2(m/n)$.
\end{lemma}
\begin{proof}
Denote the search space $\sspace = S_1 \times S_2$, $m=|S_1|$, $n=|S_2|$. Suppose that Adversary hides the target somewhere on~the diagonal of the grid $G$, i.e. on the set 
\[\sspace_{\textup{diag}} = \left\{(x, y)\colon y  = \left\lfloor (n-1) \left(1 - \frac{x}{m-1}\right) \right\rfloor \right\}.\]
Algorithm is informed that the target is on this diagonal (an example of~such a diagonal can be seen in Figure \ref{fig:graph4}). 

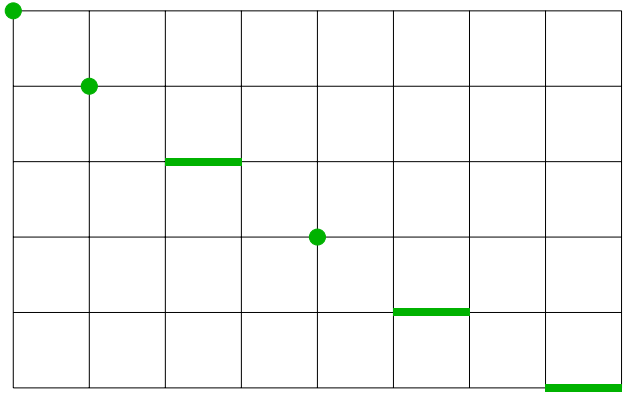
\begin{figure}[htb!]
\begin{center}
\begin{tikzpicture} 
\draw[step=1cm,black,very thin] (0,0) grid (8,5);
\filldraw [black!30!green] (0,5) circle (3pt);
\filldraw [black!30!green] (1,4) circle (3pt);
\fill[black!30!green] (2,2.95) rectangle (3,3.05);
\filldraw [black!30!green] (4,2) circle (3pt);
\fill[black!30!green] (5,0.95) rectangle (6,1.05);
\fill[black!30!green] (7,-0.05) rectangle (8,0.05);
\end{tikzpicture}
\end{center}
\caption{A diagonal on which Adversary hides the target.}
\label{fig:graph4}
\end{figure}

We define a \emph{lower} (respectively \emph{upper}) part of the grid as the set of points $(a,b)$ such that there exists $(a,y)\in \sspace_{\textup{diag}}$ that satisfies $b<y$ (respectively $b>y$).
Adversary makes the following replies to each query on a point $q$.
If $q$ belongs to the lower (respectively upper) part of the grid, then Adversary answers to the corner $(0,0)$ (respectively to the corner $(1,1)$). 
Considering the fact that Algorithm is informed that the target belongs to $\sspace_{\textup{diag}}$, such reply provides no new information to Algorithm and hence we may assume that each query selects a point $q\in \sspace_{\textup{diag}}$.
(Note that because of this artificial assumption, the game could end before the whole grid is searched, i.e., Adversary may reveal the target when the set of potential targets contains more than one element.
Such a strategy shortening factor is valid for proving a lower bound.)
If Algorithm queries a point $q$ on the diagonal, Adversary's answer eliminates only a part of the line segment that includes the~queried point, and other line segments are not influenced.
This is also achieved when Adversary always answers to one of the earlier mentioned corners $(0,0)$ and $(1,1)$.

Having described the behavior of Adversary, we now calculate the query complexity.
Note that each line segment must be searched separately because each query affects only one line segment.
Hence, the number of queries needed to search one line segment in $\sspace_{\textup{diag}}$ is at least $1 + \left\lfloor\log_2 \left\lfloor m/n \right\rfloor\right\rfloor$, because each line segment in $\sspace_{\textup{diag}}$ is of length at least $\lfloor m/n\rfloor$.
As a result, Algorithm has to search separately $n$ line segments spending at least $\log_2(m/n)$ queries on each. 
This gives the required lower bound.
\end{proof}


Our upper bound given in Lemma~\ref{lem:grid-upper} is constructive in the sense that the proof provides a search strategy.
This strategy will be recursive (hence an inductive proof) and the way to break the search space is to perform a binary search on a line segment. The formal notation for that step is the following.

Consider the search space $\sspace$ as $m \times n$ grid, 
and suppose that at some round the set of potential targets is a rectangular grid $P_k = \{m_1, \ldots, m_2\} \times \{n_1, \ldots, n_2\} \subseteq \sspace$, for some $m_1 \le m_2 < m$ and $n_1 \le n_2 < n$. For $n' = \left\lfloor\frac{n_1+n_2}{2}\right\rfloor$ we consider a horizontal line segment $L = L_0 \subseteq P_k$, containing the points $(m_1,n'),\ldots,(m_2,n')$.
We say that Algorithm performs a \emph{binary search on a segment} $L$ if the following is done.
Initially all points in $L$ are potential targets.
In each round Algorithm queries the middle element of $L$, that is, the point $(m_{\textup{mid}},n')$, where $m_{\textup{mid}}=\lfloor(m_1+m_2)/2\rfloor$.
When Adversary answers to the corner $(0,0)$ or $(0,1)$), let $L_{\textup{new}}$ be $(m_{\textup{mid}}+1,n'),\ldots,(m_2,n')$, otherwise (answer to the corner $(1,0)$ or $(1,1)$) --- let $L_{\textup{new}}$ be $(m_1,n'),\ldots,(m_{\textup{mid}}-1,n')$.
Algorithm sets $L:=L_{\textup{new}}$ and continues until $|L|=1$.
After this search, including a query on the last point on $L$, when $|L| = 1$, the set of possible targets satisfies $P_{k'} \subseteq P_k \setminus L_0$. 

Moreover, there is an important side effect that allows to reduce $P_{k'}$ even more. Let $m_{00}$ (respectively $m_{01}$) be the largest $m_{\textup{mid}}$ for which Adversary answered to the corner $(0,0)$ ($(0,1)$, respectively). When there was not such a reply set $m_{00} := m_1-1$ ($m_{01} := m_1-1$, respectively).
Analogously, let $m_{10}$ ($m_{11}$) be the smallest $m_{\textup{mid}}$ for which Adversary answered to the corner $(1,0)$ ($(1,1)$, respectively). When there was not such a reply set $m_{10} := m_2+1$ ($m_{11} := m_2+1$, respectively).  Then, (see Figure~\ref{fig:graph7PG})
\[P_{k'} = P_k \setminus R_{00} \setminus R_{10} \setminus R_{01} \setminus R_{11}, \text{~where}\]
\[R_{00} = \{0, \ldots, m_{00}\} \times \{0, \ldots, n'\}, \qquad  R_{10} = \{m_{10}, \ldots, m-1\} \times \{0, \ldots, n'\},\]
\[R_{01} = \{0, \ldots, m_{01}\} \times \{n', \ldots, n-1\}, \qquad  R_{11} = \{m_{11}, \ldots, m-1\} \times \{n', \ldots, n-1\}.\]
Note that some of the sets $R_{00}, \dots, R_{11}$ may be empty, but because a binary search on $L_0$ was performed, it holds
\[ \max \{m_{00}, m_{01}\} + 1 = \min \{m_{10}, m_{11}\},\] so its union contains the whole set $L_0$. But then the set $P_{k'}$ splits into at most two smaller \emph{byproduct grids}, that may be then searched separately, namely $P_{k'} = G_{0} \cup G_{1},$ where 
\[G_0 = \{m_{00}+1, \ldots, m_{10}-1\} \times \{0, \ldots, n'-1\},\]
\[G_1 = \{m_{01}+1, \ldots, m_{11}-1\} \times \{n'+1, \ldots, n-1\}.\] 
See Figure \ref{fig:graph7PG} for the illustration, and the proof of Lemma~\ref{lem:grid-upper} for the properties of these grids. 

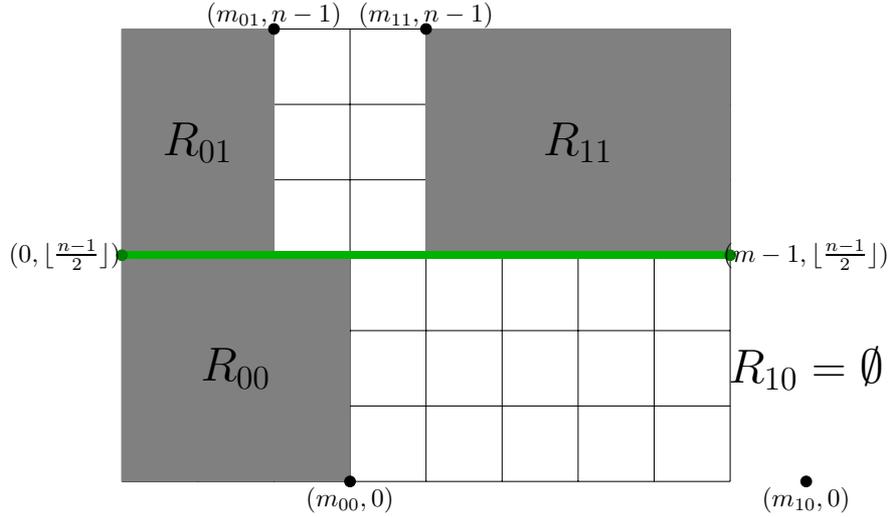
\begin{figure}[htb!]
\begin{center}
\begin{tikzpicture} 
\draw[step=1cm,black,very thin] (0,0) grid (8,6);
\fill[gray, opacity=0.7] (4,3) rectangle (8,6);
\fill[gray, opacity=0.7] (0,0) rectangle (3,3);
\fill[gray, opacity=0.7] (0,3) rectangle (2,6);
\fill[black!30!green] (0,2.95) rectangle (8,3.05);
\filldraw [black] (9,0) circle (2pt);
\filldraw [black] (3,0) circle (2pt);
\filldraw [black] (2,6) circle (2pt);
\filldraw [black] (4,6) circle (2pt);
\filldraw [black!50!green] (0,3) circle (2pt);
\filldraw [black!50!green] (8,3) circle (2pt);
  {\scalefont{0.8}
  \draw (4,6.18) node{$(m_{11},n-1)$};
  \draw (2,6.18) node{$(m_{01},n-1)$};
  \draw (3,-0.24) node{$(m_{00},0)$};
  \draw (9,-0.24) node{$(m_{10},0)$};
  \draw (-0.75,3) node{$(0,\lfloor\frac{n-1}{2}\rfloor)$};
  \draw (9,3) node{$(m-1,\lfloor\frac{n-1}{2}\rfloor)$};
 }
 {\scalefont{1.5}
  \draw (1.5,1.5) node{$R_{00}$};
  \draw (1,4.5) node{$R_{01}$};
  \draw (6,4.5) node{$R_{11}$};
  \draw (9,1.5) node{$R_{10} = \emptyset$};
 }
\end{tikzpicture}
\end{center}
\caption{Possible configuration after the binary search (with $m_{10} = m$).}
\label{fig:graph7PG}
\end{figure}

The following lemma will be proved by induction and for this to be correct we do not assume any relation between the sizes of dimensions $m$ and $n$.
\begin{lemma} \label{lem:grid-upper}
For any $m \geq 1$ and $n = 2^k - 1$, where $k \in \mathbb{N}$,
\[
\qc(m, n) \leq n \left( \log_2 \frac{m+n}{n + 1} + 3 \right) - \log_2 (n + 1).
\]
\end{lemma}

\begin{proof}
The lemma holds for $n = 1$, as $\qc(m,1) \le  \log_2 m + 1$ is the proper bound for the binary search and $\log_2 m + 1 \le \left( \log_2 \frac{m+1}{2} + 3 \right) - \log_2 (2)$. The lemma also holds for $m = 1$, as $\qc(1,n) \leq \log_2 n + 1 \le 3n - \log_2 n$.
Let now $m > 1$ and $n = 2^k - 1$, where $k \in \mathbb{N}$.
Assume for induction that for any $m_1 \in \mathbb{N}$ and $n_1 = 2^{k_1} - 1$ for ${k_1} < k$ there is
\begin{equation} \label{eq:grid-induction}
\qc(m_1, n_1) \leq n_1 \left( \log_2 \frac{m_1+n_1}{n_1 + 1} + 3 \right) - \log_2 (n_1 + 1).
\end{equation}

Performing the binary search on the line segment $L= (0, \frac{n-1}{2}) \ldots (m-1, \frac{n-1}{2})$ Algorithm obtains 
two byproduct grids:
\begin{align*}
G_0 & = \{m_{00}+1, m_{00}+2, \dots, m_{10}-1\} \times \left\{0, 1, \dots, \frac{n-1}{2}-1 \right\} \textrm{~and}\\
G_1 & = \{m_{01}+1, m_{01}+2, \dots, m_{11}-1\} \times \left\{\frac{n-1}{2}+1, \frac{n-1}{2}+2, \dots, n-1 \right\}.  
\end{align*}

Note that if $m_{00}\geq 0$ and $m_{01}\geq 0$, then Adversary provided at least one answer to the corner $(0,0)$ and one to the corner $(0,1)$.
Hence, if $m_{00}\geq m_{01}$, then changing each reply of Adversary from the corner $(0,1)$ to $(0,0)$ gives that $G_0$ remains the same and $G_1$ becomes a superset of the former $G_1$.
Due to symmetry we can hence assume that $m_{01}=-1$.
By the same argument either $m_{10}=m$ or $m_{11}=m$.
As already mentioned there is $\max \{m_{00}, m_{01}\} + 1 = \min \{m_{10}, m_{11}\}$, hence the total width of these two grids is $m$. If one of them is empty, we may split the other one arbitrarily into $G_0$ and $G_1$. Therefore, we may assume that after the binary search on the line segment $L$ the set of possible targets consists of two grids of dimensions $m' \times \frac{n-1}{2}$ and $m - m' \times \frac{n-1}{2}$.

Since for the binary search $\left\lfloor \log_2 m \right\rfloor + 1$ queries were performed and using inductive assumption in \eqref{eq:grid-induction} we have that
\begin{eqnarray*}
\qc(m, n) & \leq & \frac{n - 1}{2} \left( \log_2 \frac{m'+ \frac{n - 1}{2}}{\frac{n - 1}{2} + 1} + 3 \right) + \frac{n - 1}{2} \left( \log_2 \frac{m - m'+\frac{n - 1}{2}}{\frac{n - 1}{2} + 1} + 3 \right) \\
&&- 2\log_2 \left( \frac{n - 1}{2} + 1 \right)+\log_2 m + 1\\
&\leq & \frac{n - 1}{2} \left( \log_2 (m'+n/2)+\log_2 (m - m' + n/2) - 2\log_2 \frac{n + 1}{2} + 6 \right) \\
&&- 2 \left( \log_2 (n + 1) - 1 \right)+\log_2 (m+n) + 1\\
&=&  \frac{n - 1}{2} \big( \log_2 (m'+n/2) + \log_2 (m - m'+n/2) - 2 \log_2 (n + 1) + 8 \big) \\
&&- 2 \log_2 (n + 1) + \log_2 (m+n) + 3.
\end{eqnarray*}
By Theorem \ref{sum_of_logs}, $\log_2 (m'+n/2) + \log_2 (m - m'+n/2)\leq 2\log_2(m+n)-2$ and hence
\begin{align*}
\qc(m, n) &\leq n\log_2 (m + n) + \frac{n - 1}{2} \left( - 2 \log_2 (n + 1) + 6 \right)- 2 \log_2 (n + 1) + 3 \\
&= n\log_2 (m+n) + (n - 1) \big(- \log_2 (n + 1) + 3 \big) - 2 \log_2 (n + 1) + 3 \\
&= n\log_2 (m+n) - n \log_2(n + 1) + 3(n-1) - \log_2 (n + 1) + 3\\
&= n\log_2 \frac{m+n}{n + 1} + 3n - \log_2 (n + 1) ,
\end{align*}
which gives the bound in the lemma.
\end{proof}

\begin{proof}[Proof of Theorem~\ref{thm:grid-tight}]
For any $n$, Algorithm can consider playing on a larger grid with dimensions $m \times n_1$, where $k_1$ is minimum so that $n_1 = 2^{k_1} - 1 \ge n$.
By Lemma \ref{lem:grid-upper},
\[
\qc(m,n_1) \leq n_1 \left( \log_2 \frac{m+n_1}{n_1 + 1} + 3 \right).
\]
Consequently for all $m$ and~$n$ such that $m \geq n \geq 1$
\[
\qc(m, n) \leq 2 n \left( \log_2 \frac{m+n}{n + 1} + 3 \right) \leq 2 n \left( \log_2 \frac{m}{n + 1} + 4 \right).
\]
\end{proof}

\section{Lower bounds} \label{sec:lower}

In this section we prove lower bounds for the query complexity $\qc(n_1, n_2, n_3)$ of any 3-dimensional grid and for the query complexity $\qc(n^{\times d})$ of a $d$-dimensional cube. The former meets the upper bound shown in the next section, for the latter there is a gap of order $d$. 

\subsection{Lower bound for $3$-dimensional grid}

This section consists of a proof of the lower bound in Theorem~\ref{thm:lower3D}, which also gives the following:
\begin{corollary} \label{dolna3_szescian}
For any $n\geq 1$,
$\qc(n, n, n) \geq \frac{1}{2} n (n-1)$.
\end{corollary}

To prove Theorem~\ref{thm:lower3D},
we show a particular Adversary's strategy that guarantees the sufficient number of queries. 
Let $\sspace$ be the search space $S_1\times S_2\times S_3$, where $S_i=\{0,\ldots,n_i-1\}$, $i\in\{1,2,3\}$.
Let us assume that Adversary informs Algorithm that the~target is hidden somewhere on the (discrete) plane $H$ defined by the equation
\[H = \left \{(x_1, x_2, x_3) \in \sspace \colon  x_2 = \left\lfloor \big(n_2 - 1\big) \left( 1 - \frac{x_1}{2(n_1 - 1)} - \frac{x_3}{2(n_3 - 1)} \right) \right\rfloor\right\}.\]
Figure \ref{fig:graph21} is a graphical visualization of this plane on an exemplary grid with dimensions $6 \times 5 \times 4$.
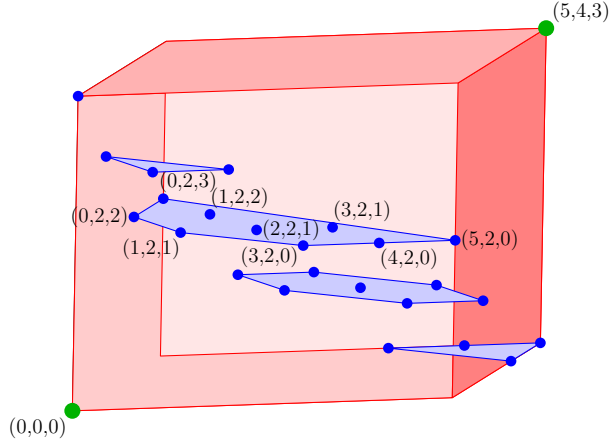
\begin{figure}[htb!]
\begin{center}
\begin{tikzpicture} [line join=round,rotate around y=90, rotate around x=-2, rotate around z=-8]
\coordinate (O) at (0,0,0);
\coordinate (A) at (0,4,0);
\coordinate (B) at (0,4,5);
\coordinate (C) at (0,0,5);
\coordinate (D) at (3,0,0);
\coordinate (E) at (3,4,0);
\coordinate (F) at (3,4,5);
\coordinate (G) at (3,0,5);
\coordinate (H) at (1,3,0);
\coordinate (I) at (0,3,1);
\coordinate (J) at (0,3,2);
\coordinate (K) at (3,2,0);
\coordinate (L) at (2,2,0);
\coordinate (M) at (2,2,1);
\coordinate (N) at (1,2,1);
\coordinate (P) at (1,2,2);
\coordinate (Q) at (1,2,3);
\coordinate (R) at (0,2,3);
\coordinate (S) at (0,2,4);
\coordinate (T) at (0,2,5);
\coordinate (U) at (3,1,1);
\coordinate (V) at (3,1,2);
\coordinate (W) at (2,1,2);
\coordinate (X) at (2,1,3);
\coordinate (Y) at (2,1,4);
\coordinate (Z) at (1,1,4);
\coordinate (AA) at (1,1,5);
\coordinate (AB) at (3,0,3);
\coordinate (AC) at (3,0,4);
\coordinate (AD) at (2,0,5);

\draw[red,fill=red!20,opacity=0.2] (O) -- (C) -- (G) -- (D) -- cycle;
\draw[black!80!red,fill=red!50,opacity=0.2] (O) -- (A) -- (E) -- (D) -- cycle;
\draw[red,fill=red!20,opacity=0.2] (O) -- (A) -- (B) -- (C) -- cycle;
\draw[red,fill=red!10,opacity=0.2] (D) -- (E) -- (F) -- (G) -- cycle;
\draw[red,fill=red!50,opacity=0.2] (C) -- (B) -- (F) -- (G) -- cycle;
\draw[red,fill=red!30,opacity=0.2] (A) -- (B) -- (F) -- (E) -- cycle;

\draw[blue,fill=blue!20,opacity=0.8] (H) -- (I) -- (J) -- cycle;
\draw[blue,fill=blue!20,opacity=0.8] (K) -- (L) -- (N) -- (R) -- (S) -- (T) -- cycle;
\draw[blue,fill=blue!20,opacity=0.8] (U) -- (W) -- (Z) -- (AA) -- (Y) -- (V) -- cycle;
\draw[blue,fill=blue!20,opacity=0.8] (AB) -- (AC) -- (G) -- (AD) -- cycle;

\node [circle, black!30!green, minimum size=6pt, inner sep=0pt, fill, label={[label distance=-2mm ,scale=0.65]200:{(0,0,0)}}] at (O) {};
\node [circle, black!30!green, minimum size=6pt, inner sep=0pt, fill, label={[label distance=-2mm ,scale=0.65]30:{(5,4,3)}}] at (F) {};

\foreach \xy in {A,H,I,J,U,V,W,X,Y,Z,AA,AB,AC,AD,G}{
\node [circle, blue, minimum size=4pt, inner sep=0pt, fill] at (\xy) {};
}

\node [circle, blue, minimum size=4pt, inner sep=0pt, fill, label={[label distance=-3.5mm ,scale=0.65]5:{(0,2,3)}}] at (K) {};
\node [circle, blue, minimum size=4pt, inner sep=0pt, fill, label={[label distance=-1.5mm ,scale=0.65]180:{(0,2,2)}}] at (L) {};
\node [circle, blue, minimum size=4pt, inner sep=0pt, fill, label={[label distance=-2.5mm ,scale=0.65]5:{(1,2,2)}}] at (M) {};
\node [circle, blue, minimum size=4pt, inner sep=0pt, fill, label={[label distance=-2.5mm ,scale=0.65]190:{(1,2,1)}}] at (N) {};
\node [circle, blue, minimum size=4pt, inner sep=0pt, fill, label={[label distance=-1.5mm ,scale=0.65]0:{(2,2,1)}}] at (P) {};
\node [circle, blue, minimum size=4pt, inner sep=0pt, fill, label={[label distance=-2.5mm ,scale=0.65]5:{(3,2,1)}}] at (Q) {};
\node [circle, blue, minimum size=4pt, inner sep=0pt, fill, label={[label distance=-2.5mm ,scale=0.65]250:{(3,2,0)}}] at (R) {};
\node [circle, blue, minimum size=4pt, inner sep=0pt, fill, label={[label distance=-2.5mm ,scale=0.65]350:{(4,2,0)}}] at (S) {};
\node [circle, blue, minimum size=4pt, inner sep=0pt, fill, label={[label distance=-1.5mm ,scale=0.65]0:{(5,2,0)}}] at (T) {};
\end{tikzpicture}
\end{center}
\caption{Points that have to be searched to find the target on the grid with dimensions $6 \times 5 \times 4$.}
\label{fig:graph21}
\end{figure}

Figure \ref{fig:graph22} is a view from the top of this grid where each level is colored differently.
\begin{figure}[htb!]
\begin{center}
\begin{tikzpicture} 
\draw[step=1cm,black,very thin] (0,0) grid (5,3);
\filldraw [black!10!yellow] (0,0) circle (3pt);
\filldraw [black!20!red] (1,0) circle (3pt);
\filldraw [black!20!red] (0,1) circle (3pt);
\filldraw [black!20!red] (2,0) circle (3pt);
\filldraw [black!20!green] (0,3) circle (3pt);
\filldraw [black!20!green] (0,2) circle (3pt);
\filldraw [black!20!green] (1,2) circle (3pt);
\filldraw [black!20!green] (1,1) circle (3pt);
\filldraw [black!20!green] (2,1) circle (3pt);
\filldraw [black!20!green] (3,1) circle (3pt);
\filldraw [black!20!green] (3,0) circle (3pt);
\filldraw [black!20!green] (4,0) circle (3pt);
\filldraw [black!20!green] (5,0) circle (3pt);
\filldraw [black!10!orange] (1,3) circle (3pt);
\filldraw [black!10!orange] (2,3) circle (3pt);
\filldraw [black!10!orange] (2,2) circle (3pt);
\filldraw [black!10!orange] (3,2) circle (3pt);
\filldraw [black!10!orange] (4,2) circle (3pt);
\filldraw [black!10!orange] (4,1) circle (3pt);
\filldraw [black!10!orange] (5,1) circle (3pt);
\filldraw [black!10!blue] (3,3) circle (3pt);
\filldraw [black!10!blue] (4,3) circle (3pt);
\filldraw [black!10!blue] (5,3) circle (3pt);
\filldraw [black!10!blue] (5,2) circle (3pt);
\end{tikzpicture}
\end{center}
\caption{A view from the top of the grid with dimensions $6 \times 5 \times 4$ (each color is a different level of the grid).}
\label{fig:graph22}
\end{figure}
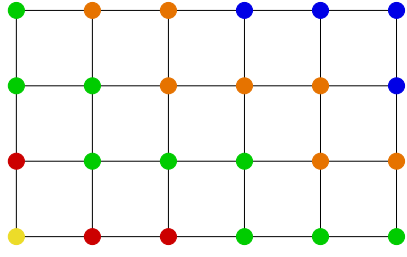

For a given $(x_1, x_2, x_3) \in H$ fix $x_2, x_3$.
Then the number of points on one level and in one row parallel to the longest side of~the~cuboid created by the grid $\sspace$ is the cardinality of the set
\[B = \left\{ x_1: \left\lfloor (n_2 - 1) \left( 1 - \frac{x_1}{2 (n_1 - 1)} - \frac{x_3}{2 (n_3 - 1)} \right) \right\rfloor = x_2 \right\}.\]
The elements of this set satisfy the two following inequalities:
\[x_2 \leq ( n_2 - 1 ) \left( 1 - \frac{x_1}{2 (n_1 - 1)} - \frac{x_3}{2 (n_3 - 1)} \right) < x_2 + 1.\]
From this we obtain bounds for $x_1$:
\[\frac{x_2}{n_2 - 1} \leq 1 - \frac{x_1}{2 (n_1 - 1)} - \frac{x_3}{2 (n_3 - 1)} < \frac{x_2 + 1}{n_2 - 1},\]
\[\frac{x_2}{n_2 - 1} - 1 + \frac{x_3}{2 (n_3 - 1)} \leq - \frac{x_1}{2 (n_1 - 1)} < \frac{x_2 + 1}{n_2 - 1} - 1 + \frac{x_3}{2 (n_3 - 1)},\]
\[1 - \frac{x_2}{n_2 - 1} - \frac{x_3}{2 (n_3 - 1)} \geq \frac{x_1}{2 (n_1 - 1)} > 1 - \frac{x_2 + 1}{n_2 - 1} - \frac{x_3}{2 (n_3 - 1)},\]
\[2 (n_1 - 1) \left( 1 - \frac{x_2}{n_2 - 1} - \frac{x_3}{2 (n_3 - 1)} \right) \geq x_1 > 2(n_1 - 1) \left( 1 - \frac{x_2 + 1}{n_2 - 1} - \frac{x_3}{2 (n_3 - 1)} \right).\]
Hence, the number of elements in the set $B$ is:
\begin{align*}
|B| &\leq 2(n_1 - 1) \left( 1 - \frac{x_2}{n_2 - 1} - \frac{x_3}{2 (n_3 - 1)} \right) - 2(n_1 - 1) \left( 1 - \frac{x_2 + 1}{n_2 - 1} - \frac{x_3}{2 (n_3 - 1)} \right)\\
{}&= 2( n_1 - 1) \left( \frac{x_2 + 1}{n_2 - 1} - \frac{x_2}{n_2 - 1} \right)\\
{}&= 2 \cdot \frac{n_1 - 1}{n_2 - 1}.
\end{align*}

A query anywhere outside of the selected plane does not provide any new information about the location of the target. Whereas a query on this plane can eliminate only a part of the line segment on the same level as the queried point and in the same row. Thus, each row on each level has to be searched separately. The selected plane has $n_1 \, n_3$ points. Let us count how many line segments of~length $2 \cdot \frac{n_1 - 1}{n_2 - 1}$ need to be searched:
\[n_1 \, n_3\,\left(\frac{2 (n_1 - 1)}{n_2 - 1}\right)^{-1} = \frac{1}{2} \, n_1 \, n_3 \, \frac{n_2 - 1}{n_1 - 1} \geq \frac{1}{2} \left( n_2 - 1 \right) n_3,\]
where for the latter we used $\frac{n_1}{n_1-1}\geq 1$.
Finally, the number $\qc^*$ of queries needed to be asked to find the target on the selected plane $B$ is equal to
\begin{align*}
\qc^* &= \frac{1}{2} \left( n_2 - 1 \right) n_3 \Bigg( \log_2 \left( 2 \cdot \frac{n_1 - 1}{n_2 - 1} \right) + 1 \Bigg) \geq \frac{1}{2} \left( n_2 - 1 \right) n_3 \left( \log_2 \frac{n_1 - 1}{n_2 - 1} + 1 \right).
\end{align*}
As basically $\qc(n_1, n_2, n_3) \geq \qc^*$ the theorem holds.
This completes the proof of Theorem~\ref{thm:lower3D}.

\subsection{Lower bound for a $d$-dimensional cube}

Briefly, our approach to prove the lower bound is to allow Adversary to answer only to corners $(0,0, \dots, 0)$ and $(1, 1, \dots, 1)$ and to restrict the target area to the diagonal hyperplane, that is antipodal to these corners. However, for the discrete hyperplane of dimension $d \ge 4$, this approach does not guarantee anymore to remove one point only from the hyperplane per one query.

\begin{proof}[Proof of Theorem~\ref{thm:lower-cube}]
Given $n \in \N^+$ let $\sspace = \{0,\ldots,n-1\}^{\times d}$ be the search space.
Let us assume that Adversary informs Algorithm that the~target is hidden somewhere on the (discrete) hyperplane $H$ defined by the equation
\[H = \left \{(x_1, \dots, x_d) \in \sspace \colon x_d = \left\lfloor n - 1 - \frac{x_1 + x_2 + \dots + x_{d - 1}}{d - 1} \right\rfloor\right\}.\]
Clearly, $|H| = n^{d-1}.$
For a given $(x_1, \dots x_d) \in H$ fix $x_2, \dots, x_d$ and calculate the number of points on a fixed level of $H$ lying on one line, that means the points that differs on the first coordinate only. 
It is defined by the cardinality of the set
\[D = \left\{ x_1: \bigg\lfloor n - 1 - \frac{x_1 + x_2 + \dots + x_{d - 1}}{d - 1} \bigg\rfloor = x_d \right\}.\]
The elements of this set satisfy the following inequalities:
\[
x_d \leq n - 1 - \frac{x_1 + x_2 + \dots + x_{d - 1}}{d - 1} < x_d + 1.
\]
This gives
\[
n - x_d - 1 \geq \frac{x_1 + x_2 + \dots + x_{d - 1}}{d - 1} > n - x_d - 2,
\]
\[
( d - 1)( n - x_d - 1) \geq x_1 + x_2 + \dots + x_{d - 1} > ( d - 1)( n - x_d - 2),
\]
and hence
\[
( d - 1)(n - x_d - 1) - x_2 - \dots - x_{d - 1} \geq x_1 > (d - 1)( n - x_d - 2) - x_2 - \dots - x_{d - 1}.
\]
Hence, the size of the set $D$ is bounded by 
\begin{align*}
|D| & \leq \left( d - 1 \right) \left( n - x_d - 1 \right) - \left( d - 1 \right) \left( n - x_d - 2 \right) \\
{}&= \left( d - 1 \right) \left( n - x_d - 1 \right) - \left( d - 1 \right) \left( n - x_d - 1 \right) + d - 1 = d - 1.
\end{align*}
It means that in each line on one level of the grid $H$, there are $d - 1$ points except for the line segments on the edges, where they can be shorter. Note that, although the calculation was done for the first coordinate, by the symmetry of $H$ the same holds for any other coordinate.

Adversary informs Algorithm that he will always answer for a query on a point $x = (x_1, \dots, x_d)$ in such a~way that:
\begin{itemize}
    \item if $x \notin H$, then he answers in such a way to not eliminate any part of the hyperplane $H$,
    \item otherwise, calculating $s = \sum_{i=1}^d x_i$, if the remainder $l$ of dividing $s$ by $d - 1$, $l=s\mod(d-1)$, is $0$ or $l$ belongs to the interval $\left( \frac{d - 1}{2}, d - 1 \right]$, then Adversary answers to the corner $(1,\ldots,1)$
    \item  if $l$ belongs to the interval $\left( 0, \frac{d - 1}{2} \right]$, then Adversary answers to the corner $(0,\ldots,0)$.
\end{itemize}

It means that Algorithm needs to query all points on the hyperplane $H$ for~which the remainder of dividing the sum of its coordinates by $d - 1$ is $\left\lfloor \frac{d - 1}{2} \right\rfloor$ or $\left\lfloor \frac{d - 1}{2} \right\rfloor + 1$. As a result, $2 \left \lfloor \frac{n^{d - 1}}{d - 1} \right \rfloor$ points have to be queried.
So, for Adversary's optimal strategy 
$\qc(H) \geq 2 \left \lfloor \frac{n^{d - 1}}{d - 1} \right \rfloor$.
\end{proof}

\section{Upper bound}

We finally prove the upper bound for the query complexity on a $d$-dimensional grid stated in Theorem~\ref{thm:upper-d}.


\begin{proof}[Proof of Theorem~\ref{thm:upper-d}]
The proof is done by induction on $d$ with base step (for $d = 2$) given by Theorem~\ref{thm:grid-tight}.
For the inductive step, fix $d \geq 3$ and assume that the theorem holds true for all grids on dimension less that $d$.
Given $n_1 \geq n_2 \geq \cdots \geq n_d\geq 1$ let $G$ be a grid with dimensions $n_1 \times n_2 \times \dots \times n_d$. Algorithm treats the game on $G$ like $n_d$ independent games on $(d-1)$-dimensional grids $G'$ with dimensions $n_1 \times n_2 \times \dots \times n_{d-1}$ each.
Thus, Algorithm plays $n_d$ independent games, each with an upper bound on the query complexity being according to the inductive assumption:
\[\qc(G')\in O \left( \prod_{i = 2}^{d-1} n_i \left( \log_2 \frac{n_1}{n_2} + 1 \right) \right).\]
Since $\qc(G) \leq \, n_d\cdot\qc(G')$, the theorem holds by induction. 
\end{proof}

\bibliographystyle{plain}
\bibliography{references}

\begin{thebibliography}{10}

\bibitem{Ben-AsherFN99}
Yosi Ben{-}Asher, Eitan Farchi, and Ilan Newman.
\newblock Optimal search in trees.
\newblock {\em {SIAM} J. Comput.}, 28(6):2090--2102, 1999.

\bibitem{Ben-OrH08}
Michael Ben{-}Or and Avinatan Hassidim.
\newblock The bayesian learner is optimal for noisy binary search (and pretty
  good for quantum as well).
\newblock In {\em FOCS}, pages 221--230, 2008.

\bibitem{BoczkowskiFKR21}
Lucas Boczkowski, Uriel Feige, Amos Korman, and Yoav Rodeh.
\newblock Navigating in trees with permanently noisy advice.
\newblock {\em {ACM} Trans. Algorithms}, 17(2):15:1--15:27, 2021.

\bibitem{BonisGV97}
Annalisa~De Bonis, Luisa Gargano, and Ugo Vaccaro.
\newblock Group testing with unreliable tests.
\newblock {\em Inf. Sci.}, 96(1{\&}2):1--14, 1997.

\bibitem{ChangH80}
Gerard~J. Chang and Frank~K. Hwang.
\newblock A group testing problem.
\newblock {\em {SIAM} J. Algebraic Discret. Methods}, 1(1):21--24, 1980.

\bibitem{CicaleseJLV12}
Ferdinando Cicalese, Tobias Jacobs, Eduardo~Sany Laber, and Caio~Dias Valentim.
\newblock The binary identification problem for weighted trees.
\newblock {\em Theor. Comput. Sci.}, 459:100--112, 2012.

\bibitem{CicaleseKLPV16}
Ferdinando Cicalese, Bal{\'{a}}zs Keszegh, Bernard Lidick{\'{y}},
  D{\"{o}}m{\"{o}}t{\"{o}}r P{\'{a}}lv{\"{o}}lgyi, and Tom{\'{a}}\v{s} Valla.
\newblock On the tree search problem with non-uniform costs.
\newblock {\em Theor. Comput. Sci.}, 647:22--32, 2016.

\bibitem{DeligkasMS19}
Argyrios Deligkas, George~B. Mertzios, and Paul~G. Spirakis.
\newblock Binary search in graphs revisited.
\newblock {\em Algorithmica}, 81(5):1757--1780, 2019.

\bibitem{Deppe2007}
Christian Deppe.
\newblock {\em Coding with Feedback and Searching with Lies}, pages 27--70.
\newblock Springer Berlin Heidelberg, Berlin, Heidelberg, 2007.

\bibitem{Dereniowski08}
Dariusz Dereniowski.
\newblock Edge ranking and searching in partial orders.
\newblock {\em Discret. Appl. Math.}, 156(13):2493--2500, 2008.

\bibitem{DereniowskiGP23}
Dariusz Dereniowski, Przemyslaw Gordinowicz, and Pawel Pralat.
\newblock Edge and pair queries-random graphs and complexity.
\newblock {\em Electron. J. Comb.}, 30(2), 2023.

\bibitem{DereniowskiKUZ17}
Dariusz Dereniowski, Adrian Kosowski, Przemys\l{}aw Uzna\'{n}ski, and Mengchuan
  Zou.
\newblock Approximation strategies for generalized binary search in weighted
  trees.
\newblock In {\em {ICALP}}, pages 84:1--84:14, 2017.

\bibitem{DereniowskiLU21}
Dariusz Dereniowski, Aleksander Lukasiewicz, and Przemyslaw Uznanski.
\newblock An efficient noisy binary search in graphs via median approximation.
\newblock In {\em {IWOCA}}, volume 12757, pages 265--281, 2021.

\bibitem{DLU21}
Dariusz Dereniowski, Aleksander Lukasiewicz, and Przemyslaw Uznanski.
\newblock Noisy searching: simple, fast and correct.
\newblock {\em CoRR}, abs/2107.05753, 2021.

\bibitem{DyachkovVPS16}
Arkady~G. D'yachkov, Ilya~V. Vorobyev, N.~A. Polyanskii, and Vladislav~Yu.
  Shchukin.
\newblock On a hypergraph approach to multistage group testing problems.
\newblock In {\em {ISIT}}, pages 1183--1191. {IEEE}, 2016.

\bibitem{Emamjomeh-ZadehK17}
Ehsan Emamjomeh{-}Zadeh and David Kempe.
\newblock A general framework for robust interactive learning.
\newblock In {\em {NIPS}}, pages 7085--7094, 2017.

\bibitem{EKMS20}
Ehsan Emamjomeh-Zadeh, David Kempe, Mohammad Mahdian, and Robert~E. Schapire.
\newblock Interactive learning of a dynamic structure.
\newblock In {\em {ALT}}, pages 277--296, 2020.

\bibitem{EKS16}
Ehsan Emamjomeh{-}Zadeh, David Kempe, and Vikrant Singhal.
\newblock Deterministic and probabilistic binary search in graphs.
\newblock In {\em {STOC}}, pages 519--532, 2016.

\bibitem{FeigeRPU94}
Uriel Feige, Prabhakar Raghavan, David Peleg, and Eli Upfal.
\newblock Computing with noisy information.
\newblock {\em {SIAM} J. Comput.}, 23(5):1001--1018, 1994.

\bibitem{KarpK07}
Richard~M. Karp and Robert Kleinberg.
\newblock Noisy binary search and its applications.
\newblock In {\em SODA}, pages 881--890, 2007.

\bibitem{LamY01}
Tak~Wah Lam and Fung~Ling Yue.
\newblock Optimal edge ranking of trees in linear time.
\newblock {\em Algorithmica}, 30(1):12--33, 2001.

\bibitem{MozesOW08}
Shay Mozes, Krzysztof Onak, and Oren Weimann.
\newblock Finding an optimal tree searching strategy in linear time.
\newblock In {\em SODA}, pages 1096--1105, 2008.

\bibitem{NesetrilM06}
Jaroslav Ne\v{s}et\v{r}il and Patrice~Ossona de~Mendez.
\newblock Tree-depth, subgraph coloring and homomorphism bounds.
\newblock {\em Eur. J. Comb.}, 27(6):1022--1041, 2006.

\bibitem{OnakP06}
Krzysztof Onak and Pawel Parys.
\newblock Generalization of binary search: Searching in trees and forest-like
  partial orders.
\newblock In {\em FOCS}, pages 379--388, 2006.

\bibitem{Pelc02}
Andrzej Pelc.
\newblock Searching games with errors---fifty years of coping with liars.
\newblock {\em Theoretical Computer Science}, 270(1):71 -- 109, 2002.

\end{thebibliography}
\end{document}